\documentclass[submission,copyright,creativecommons]{eptcs}
\providecommand{\event}{AFL 2023} 

\usepackage{amsmath}
\usepackage{
amssymb}
\usepackage{graphicx}

\ifpdf
  \usepackage{underscore}         
  \usepackage[T1]{fontenc}        
\else
  \usepackage{breakurl}           
\fi

\newtheorem{theorem}{Theorem}
\newtheorem{lemma}{Lemma}
\newtheorem{example}{Example}

\newtheorem{corollary}{Corollary}

\newtheorem{proposition}{Proposition}

\newenvironment{proof}{\begin{trivlist}
                       \item[]{\bf Proof}
                       \hspace{0cm} }{\hfill {\large $\bullet$}
                       \end{trivlist}}

\newcommand{\be}{\begin{eqnarray}}
\newcommand{\ee}{\end{eqnarray}}

\newcommand{\ignore}[1]{}

\title{State-deterministic Finite Automata with Translucent Letters and
Finite Automata with Nondeterministically Translucent Letters}
\author{Benedek Nagy
\institute{Department of Mathematics, Faculty of Arts and Sciences\\
Eastern Mediterranean University\\
99628 Famagusta, North Cyprus, Mersin-10, Turkey\\ and \\
Department of Computer Science, Institute of Mathematics and Informatics,\\ Eszterh\'azy K\'aroly Catholic University\\ Eger, Hungary}
\email{nbenedek.inf@gmail.com}
}

\def\titlerunning{Nondeterministically translucent letters and state-determinism}
\def\authorrunning{Benedek Nagy}
\begin{document}
\maketitle

\begin{abstract}
Deterministic and nondeterministic finite automata with translucent letters were introduced by Nagy and Otto more than a decade ago as Cooperative Distributed systems of a kind of stateless restarting automata with window size one. These finite state machines have a surprisingly large expressive power: all commutative semi-linear languages and all rational trace languages can be accepted by them including various not context-free languages. While the nondeterministic variant defines a language class with nice closure properties, the deterministic variant is weaker, however it contains all regular languages, some non-regular context-free languages, as the Dyck language, and also some languages that are not even context-free. In all those models for each state, the letters of the alphabet could be in one of the following categories: the automaton cannot see the letter (it is translucent), there is a transition defined on the letter (maybe more than one transition in nondeterministic case) or none of the above categories (the automaton gets stuck by seeing this letter at the given state and this computation is not accepting).

State-deterministic automata are recent models, where the next state of
the computation determined by the structure of the automata and it is independent of the processed letters. In this paper our aim is twofold, on the one hand, we investigate
state-deterministic  finite automata with translucent letters. These automata are specially restricted deterministic  finite automata with translucent letters.

In the other novel model we present, it is allowed that for a state the set of translucent letters and the set of letters for which transition is defined are not disjoint. One can interpret this fact that the automaton has a nondeterministic choice for each occurrence of such letters to see them (and then erase and make the transition) or not to see that occurrence at that time. Based on these semi-translucent letters, the expressive power of the automata increases, i.e., in this way a proper generalization of the previous models is obtained.
\end{abstract}

\section{Introduction}

The history of automata with translucent letters has begun using the technical name
cooperative distributed systems of stateless restarting automata with window size one \cite{LATAtrace}, while the term finite state acceptors with translucent letters appeared in \cite{BILC} reinterpreting the aforementioned technical name.
Basically (formal definitions will be recalled in Section~\ref{sec:prelim}), in a finite automaton with translucent letters, in each state some of the letters of the alphabet are translucent, and the automaton sees the first occurrence of a non-translucent letter (after the occurrences of translucent letters in the prefix of the remaining input in the given configuration, if any) and if there is a transition defined on this letter (say, the letter is readable) in the actual state, after erasing this letter, the next state is chosen according to the transition function, and the computation continues. It may happen that there are only translucent letters (or no letters at all) in the remained input, then the computation is accepting if the actual state is a final state. Automata with translucent letters can be applied in linguistics \cite{NaKo14}, and also modelling various trace languages used to describe parallel processes \cite{DiRo95,JanickiKKM19,MateescuSY00,IWCIA17} based on commutations and partial commutations \cite{CartierFoata69,Benedek09}.

In fact, there are various models in automata theory where the processing on the input may not go strictly left to right. One of these models is the restarting automata which is developed for linguistical motivation to do analysis by reduction: in a nutshell, these automata have a read-write window and they are searching for some specific pattern in the window to reduce, i.e., shorten its content, and then they are restarting the computation on the new content of the tape. It may also happen that the automaton accepts based on what is in its window. Interested readers may be referred to \cite{restart,restar} to see the various models, their computations, accepted languages and their properties. Restarting R automata with window size one can do only one type of reduction, to erase the letter in the window, hereby shortening the tape. Stateless deterministic variants of them are the simplest models of restarting automata. Instead of adding states to the system, their cooperative distributed systems (shortly CD systems) are developed \cite{LATAtrace,JCSS} and found to be very interesting with a surprisingly large expressive power as, e.g., they are able to accept all rational trace languages. The components of such systems play the role of the states in the reinterpreted model, in the nondeterministic finite automata with translucent letter. Two types of deterministic models of the CD systems of  restarting R automata with window size one are also studied \cite{LATA2011}: In strictly deterministic models, the next component is uniquely defined by the actual component and it does not depend on the letter being processed (i.e., erased) in the actual computation step.
However, in globally deterministic CD systems of  restarting R automata with window size one
the next component is deterministically chosen based on the actual component and on the erased letter, somewhat similarly as it is in the usual deterministic finite automata. Consequently,
this letter model is equivalent, by the reinterpretation, to the deterministic finite automata with translucent letters.

Other models not consuming the input from left to right are various 2-head models that process the input parallely from both extremes \cite{WK,linAUTO,DNA08,2-head,JLC,DNAbook}. Some of these finite state models are capable to accept exactly the linear context-free languages. Moreover, in the bio-inspired models named $5'\to 3'$ Watson-Crick finite automata, the automaton with its both heads may read strings in a computation step \cite{CiE09,NaCo-Shag,AFL-SH}.
There are various interesting concepts related to determinism introduced and studied for these $5'\to 3'$ Watson-Crick finite automata models.
   The deterministic variant where the concept of determinism fits well to the usual concept of determinism is less powerful in the sense that only a proper subset of the class of the linear context-free languages can be accepted by them. This class is called 2detLIN, and it is incomparable with the class detLIN containing the languages that are accepted by deterministic one-turn pushdown automata \cite{ActInf-Sh,UCNC-Sh}.
 The model where the next state is uniquely defined by the actual state and does not depend on what is being processed from the input in this step of computation is called state-deterministic and studied in \cite{NaCo-sd}. At $5'\to 3'$ Watson-Crick finite automata, the state-deterministic variants are very restricted, but they may do some nondeterministic computations, and thus, the language class accepted by them is incomparable with 2detLIN.
Other type of determinism, the quasi-determinism is introduced and studied in \cite{NCMA22}.
 In these automata, even if the state of the next configuration is uniquely determined by the actual configuration (actual state and remaining input), there could be more than one possible next configuration.
The quasi-deterministic $5'\to 3'$ Watson-Crick finite automata accept a superclass of languages of both the classes accepted by state-deterministic and by deterministic $5'\to 3'$ Watson-Crick finite automata.

There are other models of computations which process the input not strictly left to right, including various jumping automata \cite{JUMP-Sz,jump,5-3jumping} %
and input revolving automata \cite{input-revolving}, just to mention a few more models. These models became very popular in the last decades.
 The combination of the mentioned 2-head finite state model with translucent letters allow to accept all linear trace languages \cite{LINtrace}. Pushdown automata with translucent letters can be used to characterise context-free trace languages \cite{CFtrace}.

In this paper, on the one hand, we investigate the state-deterministic finite automata with translucent letters and give some results on the class of languages accepted by them.

On the other hand, we investigate finite automata with translucent letters by relaxing the condition that for a state the set of translucent and readable letters is disjoint. In this way, the translucency becomes nondeterministic and thus, we may expand both the deterministic and nondeterministic finite automata with translucent letters to allow nondeterministic translucency. Our other main result is that we can show that this model is a real expansion of the basic models, the class of accepted languages is a %
  superclass of the class languages of the original model.

Recently another extension of the finite automata with translucent letters was investigated in which in the computation the head is not restarting after erasing a symbol, but continues from the position where this letter has been erased \cite{Fran-Fred}. We believe that our new type of restrictions and  extensions are also giving some new interesting insights and results to this particular field of automata theory.

The structure of the paper is as follows. In the next section we recall some formal preliminaries and the basic definitions of finite automata with translucent letters.
 Section~\ref{sec:stateD} is devoted to a restricted class of deterministic models, namely, to the state-deterministic finite automata with translucent letters;
while in Section~\ref{sec:def} we present our other new concept by allowing nondeterminism based on translucency. We show that this model is more powerful than the original model, however, still only semi-linear languages can be accepted.
 Finally, conclusions close the paper.

\section{Preliminaries}\label{sec:prelim}

We assume that the reader is familiar with the basic concepts of formal languages and automata \cite{Harrison78,HopcroftUllman79}, however, to fix our notations, we formally recall some basic definitions.
We denote the empty word by $\lambda$.

We say that the languages $L_1$ and $L_2$ are letter equivalent, if for any word $x\in L_1$ we may find a word $y\in L_2$ such that $y$ is obtained from $x$ by reordering (permuting) its letters and also  for any word $x\in L_2$ we may find a word $y\in L_1$ with the same property. It is known that a language is \emph{semi-linear} if there is a regular language that is letter
equivalent with it. All context-free languages are semi-linear %
\cite{Parikh61} and there are context-sensitive languages that are not semi-linear.
We do not detail here partial commutations, commutations, traces and trace languages, interested readers may be referred to \cite{CartierFoata69,DiRo95} and for their relations to automata with translucent letters to \cite{LATAtrace,CFtrace,LINtrace}.

A \emph{nondeterministic finite automaton} (NFA)
is a pentuple $A = (Q,\Sigma,I,F,\delta)$,
where $Q$ is the finite set of internal states,
$\Sigma$ is the finite alphabet containing the input letters,
$I\subseteq Q$ is the set of initial states,
$F\subseteq Q$ is the set of final (or accepting) states,
and $\delta:Q\times \Sigma \to 2^Q$ is the transition relation.
If $|I|=1$ and $|\delta(q,a)|\le 1$ holds for all $q\in Q$ and all $a\in\Sigma$,
then $A$ is a \emph{deterministic finite automaton} (DFA).
Notice that, in general, in NFAs we allow multiple initial states, but we do not allow transitions by the empty word.

An NFA $A$ works as follows. Let an input string $w\in\Sigma^*$ be given,
then $A$ starts its computation in a state $q_0$ that is chosen
nondeterministically from the set $I$ of all initial states.
This configuration is encoded as $q_0w$ (for simplicity, we may assume that $Q\cap \Sigma = \emptyset$).
Now $A$ reads the first letter of $w$, say $a$ (let $w=au$), thereby deleting (consuming) this occurrence of letter $a$,
and it changes its internal state to a state $q_1$ that is chosen nondeterministically
from the set $\delta(q_0,a)$, formally we may write that the new configuration $q_1 u$ is reached.
However, it may happen that $\delta(q_0,a)$ is empty, then $A$ gets stuck and this computation fails in this input. Otherwise, $A$ continues the computation from the configuration $q_1 u$ by
reading the input letter by letter until either $w$ has been consumed completely or the computation fails (similarly as we have described).
We say that $A$ accepts $w$  from the initial configuration $q_0 w$ if it reaches a configuration $q_f \cdot \lambda$ in a computation starting from $q_0 w$, where  $q_f\in F$ is a final state. %
By $L(A)$ we denote the set of all strings $w\in\Sigma^*$ for which
$A$ has an accepting computation in the sense described above.

It is well-known that the class %
 of languages %
 that are accepted
by NFAs coincides with the class of regular languages, and that
DFAs accept exactly the same languages.

Now we recall the nondeterministic finite automata with translucent letters
from %
\cite{BILC}.

A \emph{finite state automaton with translucent letters} ({NFAwtl})
is defined as 
 $A = (Q,\Sigma,\$,\tau,I,F,\delta)$,
where $Q$, %
$\Sigma$, $I$ and $F$ are the same as at an NFA; %
$\$\not\in\Sigma$ is a special symbol that is used technically as an \emph{endmarker},
$\tau:Q\to 2^\Sigma$ is the \emph{translucency mapping},
and $\delta:Q\times\Sigma\to 2^Q$ is the \emph{transition relation}
that satisfies the following condition:
$$\forall q\in Q\; \forall a\in \tau(q): \delta(q,a)=\emptyset.$$
For each state $q\in Q$, the letters from the set $\tau(q)$ are translucent for~$q$,
that is, in state $q$ the automaton $A$ does not see these letters.
$A$ is called \emph{deterministic finite state automaton with translucent letters}, abbreviated as {DFAwtl},
if $|I|=1$ and
if $|\delta(q,a)|\le 1$ for all $q\in Q$ and all $a\in\Sigma$.

An {NFAwtl} $A = (Q,\Sigma,\$,\tau,I,F,\delta)$ works as follows.
Let   $w\in\Sigma^*$ be an input word. $A$
starts in a nondeterministically chosen initial state
$q\in I$ with the word $w\cdot\$$ on its input tape, that is $q_0 w \$ $ is an initial configuration.
A computation step of $A$ is defined as follows. Assume that $w= a_1a_2\cdots a_n$ for some $n\ge 1$ and $a_1,\ldots,a_n\in\Sigma$.
Then $A$ looks for the first occurrence from the left of a letter
that is not translucent (say visible) for the current state $q$,
more precisely, if $w=uav$ such that $u\in(\tau(q))^*$ and $a\not\in \tau(q)$,
then $A$ nondeterministically chooses a state $q'\in\delta(q,a)$,
erases the letter $a$ from the tape thus producing the tape contents $uv\cdot\$$,
and its internal state is set to~$q'$. Therefore after this computation step the configuration is $q' uv \$ $ and the computation continues from this configuration
by looking for the first visible letter of $uv$ at state $q'$.
However, it may happen that $\delta(q,a)=\emptyset$ for the first visible letter $a$,
$A$ halts without accepting, this computation fails.
Finally, if $w\in(\tau(q))^*$ for a configuration $q w\$ $ (including the possibility that the configuration is in fact $q \cdot \lambda \cdot \$ $),
then $A$ reaches the \$-symbol and the computation halts.
In this case $A$ accepts if $q\in F$ is a final state;
otherwise, it does not accept.
A word $w\in\Sigma^*$ is \emph{accepted by} $A$ if there exists an initial state $q_0\in I$
and an accepting computation from $q_0w\cdot \$ $.
Further, the empty word $\lambda$ is accepted by $A$ if there exists an initial state
$q_0\in Q$ such that $q_0$ is also a final state.
Now $L(A) = \{\,w\in\Sigma^*\mid w\mbox{ is accepted by }A\,\}$
is the \emph{language accepted by} $A$. Notice that the endmarker is, in fact, needless; we kept it only for traditional reason.

The classical \emph{nondeterministic finite automata} (NFA)
is obtained from the NFAwtl by removing the endmarker $\$$ and by ignoring the translucency relation~$\tau$,
and the \emph{deterministic finite-state acceptor} (DFA) is obtained from the DFAwtl in the same way.
Thus, the NFA (DFA) can be interpreted as a special type of NFAwtl (DFAwtl).
Accordingly, all regular languages are accepted by DFAwtl.
Moreover, {DFAwtl}s are much more expressive than standard DFAs
as shown by the following example.

\begin{example}\label{exa:nmkj}
Consider the DFAwtl
 $A = (Q,\Sigma,\$,\tau,I,F,\delta)$, where $Q = \{q_0,q,q_a,q_b,q_c,q_d\}$,
$I=\{q_0\}$, $F=\{q\}$, $\Sigma = \{a,b,c,d\}$,
and the functions $\tau$ and $\delta$ are defined as follows:
$$\begin{array}{rclrcl}
\tau(q_0) & = & \{b,c,d\},       & \delta(q_0,a) & = & \{q_a\}, \\
\tau(q) & = & \emptyset, \hspace{1cm}
    & \delta(q,a) & = & \{q_a\}, \ \delta(q,b) = \{q_b\}, \ \delta(q,c) = \{q_c\}, \ \delta(q,d) = \{q_d\}, \\
\tau(q_a) & = & \{a,c,d\},        & \delta(q_a,b) & = & \{q\},\\
\tau(q_b) & = & \{b,c,d\},        & \delta(q_b,a) & = & \{q\},\\
\tau(q_c) & = & \{a,b,c\},        & \delta(q_c,d) & = & \{q\},\\
\tau(q_d) & = & \{a,b,d\},        & \delta(q_d,c) & = & \{q\}.
\end{array}$$
Further, $\delta(p,x)=\emptyset$ for all other pairs $(p,x)\in Q\times\Sigma$.
Firstly, the input must have an $a$, which is consumed in the first step of the computation, then a $b$ is consumed.
   One may see that after that the automaton reads the first letter of the remaining input and depending on what it was, in the next step consumes the first occurrence of a letter that is a pair of the previously erased one, where pairs are $a$-s with $b$-s and $c$-s with $d$-s.
    Consequently $A$ accepts the language $L^{ab} = \{w\in\{a,b,c,d\}^*~|~ |w|_a = |w|_b >0 \text{ and } |w|_c = |w|_d \}$. Similarly, by permuting the roles of the letters, e.g., the language
  $L^{ac} = \{w\in\{a,b,c,d\}^*~|~ |w|_a = |w|_c >0 \text{ and } |w|_b = |w|_d \}$ is also accepted by a DFAwtl. However, the union of these two languages can be accepted by an NFAwtl, but cannot with any DFAwtl.
 This latter fact can be shown somewhat analogously to the fact that the context-free language $\{a^nb^n c^md^m\}\cup\{a^nb^mc^md^n\}$ is not deterministic context-free. We skip the formal proof because the lack of space.
\end{example}

As we have already described NFAwtl and DFAwtl are reformulations of cooperative distributed systems of stateless deterministic restarting R automata with window size one. The DFAwtl, in fact, are reinterpretations of stateless globally deterministic CD-R(1)-systems \cite{LATA2011}.

 Recently various concepts about deterministic computations have been emerged, therefore,
 we recall the concept of state-determinism from \cite{NaCo-sd}.

  An automaton is \emph{state-deterministic} if for each of its state $q\in Q$,
if there is a transition from $q$ and it goes to state $p$ (i.e., $p\in\delta(q,a)$),
 then every transition from $q$ goes to $p$, that is,
   if an automaton has state $q$ in its actual configuration, then, if the computation continues, the state of the next configuration is uniquely determined and it is $p$.

We are continuing the paper in this line.

\section{On state-deterministic finite automata with translucent letters}\label{sec:stateD}

As our first result, we investigate the state-deterministic FAwtl (SFAwtl for short) by applying this type of concept of determinism to NFAwtl.

  First, we recall the concept of \emph{stateless strictly deterministic CD-R(1)-systems} \cite{LATA2011}. In these systems there is only one initial state, and there is exactly one successor component for each component. One may think, that in the terminology of finite automata with translucent letters we can interpret it with the conditions
  $| I | = 1$ 
   and
   for each $q\in Q$, $\left| \bigcup_{a\in\Sigma} \delta(q,a)  \right| = 1$ which may lead to a very similar concept as state-determinism.
However, this is not exactly the case,
stateless strictly deterministic CD-R(1)-systems and the state-deterministic FAwtl are in close relation, but in a CD-R(1)-system one may use the computation step ``Accept'' at any component on a given non translucent letter, while in NFAwtl the acceptance condition is defined in a different way. We are showing some explicit difference of these models in this section.

We present an example to show that these restricted automata are still able to accept non trivial languages.

\begin{example}
  Let $A = (Q,\Sigma,\$,\tau,I,F,\delta)$, where $Q = \{q_0,q_1\}$,
$I=\{q_0\}=F$, $\Sigma = \{0,1\}$,
and the functions $\tau$ and $\delta$ are defined as follows:
$$\begin{array}{rclrcl}
\tau(q_0) & = & \emptyset, \hspace{2.75cm}      & \delta(q_0,0) & = & \{q_1\},\\
\tau(q_1) & = & \{0\},        & \delta(q_1,1) & = & \{q_0\}.
\end{array}$$
By observing the structure of this automaton, it is clearly state-deterministic.
Considering the accepted language, it is the Dyck language, where $0$ refers to opening and $1$ to closing brackets.
\end{example}

\begin{example}
Let $A = (Q,\Sigma,\$,\tau,I,F,\delta)$, where $Q = \{q_0,q_1,q_2,q_3\}$,
$I=\{q_0\}=F$, $\Sigma = \{a,b,c,d\}$,
and the functions $\tau$ and $\delta$ are defined as follows:
$$\begin{array}{rclrcl}
\tau(q_0) & = & \emptyset, \hspace{3.5cm}      & \delta(q_0,a) & = & \{q_1\},\\
\tau(q_1) & = & \{a,c,d\},        & \delta(q_1,b) & = & \{q_2\},\\
\tau(q_2) & = & \{a,b,d\},        & \delta(q_2,c) & = & \{q_3\},\\
\tau(q_3) & = & \{a,b,c\},        & \delta(q_3,d) & = & \{q_0\}.
\end{array}$$
Further, $\delta(q,x)=\emptyset$ for all other pairs $(q,x)\in Q\times\Sigma$.
On the one hand, it is easy to check that $A$ is a DFAwtl which is, in fact, also state-deterministic.
On the other hand, the language accepted by $A$ intersected by the regular language $a^*d^*c^*b^*$ is the non context-free language $\{a^nd^nc^nb^n~|~ n\ge 0\}$, and thus
$A$ accepts a language that is not context-free.
\end{example}

From this example, using the fact that any language accepted by NFAwtl has a letter-equivalent sublanguage that is regular \cite{LATAtrace} (but the language $\{a^nd^nc^nb^n~|~ n\ge 0\}$ does not), we can deduce that:

\begin{proposition}
  The language class accepted by state-deterministic FAwtl is not closed under intersection with regular languages.
\end{proposition}

Based on \cite{NaCo-sd}, we know that state-deterministic FAwtl have the graph structure with no branching, that is, either a line graph (starting from the sole initial state) or a line graph with an additional edge from the last state to a state.

Clearly languages like $a^*$, $b^*$, $a+aaa$, $ab+ba$ are accepted by state-deterministic FAwtl with 1, 1, 4 and 3 states, respectively. For $ab+ba$
translucency can be used in the initial state, e.g., $a$ is translucent and transition with $b$ leads to the next state, from which the computation may continue by reading an $a$ to reach the final state.

Now, %
we present a relatively simple example language that is not accepted by any SFAwtl.

\begin{example}
  The regular language described by $a^* + b^*$ is not accepted by any state-deterministic FAwtl. It is easy to see that, by assuming that an SFAwtl $A$ accepts the given language, after reading an $a$ or a $b$, $A$ must be in the same state, however, the possible computations after erasing an $a$ or erasing a $b$ must not be the same, since this would lead to accept words containing both $a$ and $b$, contradicting to our assumption on the accepted language.
\end{example}

Based on the above example, we may deduce the following closure property:

\begin{proposition}
  The language class accepted by SFAwtl is not closed under union.
\end{proposition}

 We may also summarize some hierarchy type results based on the previously shown examples.
\begin{proposition}
State-deterministic FAwtl are deterministic, i.e., they are also DFAwtl. \\
Further, the language class accepted by SFAwtl includes some non context-free languages, but on the other hand, does not include all regular languages.
\end{proposition}

We recall from \cite{IJCM} that the language class accepted by stateless strictly deterministic CD-R(1)-systems is closed under complement.
 We show that this is not the case with the state-deterministic FAwtl, thus, in this way,
 we also show that the new concept is not a reinterpretation of these CD systems.

\begin{proposition}
  The language class accepted by SFAwtl is not closed under complement.
\end{proposition}
\begin{proof}
  On the one hand, as we have described, the Dyck language over $\{0,1\}$ is accepted by state-deterministic FAwtl.
  Now, on the other hand, we show that its complement $L^c$ is not. Let as assume towards a contradiction that there is an  SFAwtl $A$ that accepts $L^c$.
  Since $\lambda$ is not in $L^c$, the initial state $q_0$ is not a final state of $A$.
Let us consider the cases by seeing which of the letters could be translucent at $q_0$.
\begin{itemize}
\item Clearly, it cannot happen that both $0$ and $1$ are translucent, since then, no words would be accepted.
\item In case either $0$ or $1$ is translucent, there must be a transition with the other letter from $q_0$, to another state, say state $q_1$. Now, on the one hand, the input $01$ should not be accepted, but the input $10$ should be. However, in this case, both of these inputs lead to the same configuration after the first step of the computation. This leads to a contradiction, since from here either both of them are accepted by $A$, or none of them.
 \item The last possibility is when there are no translucent letters at $q_0$. Since $A$ must accept words starting with $0$, e.g., $0,00,000,001,010$ and also words starting with $1$, e.g., $1,10,11,100$, the transition with both $0$ and $1$ must go an accepting state $q_1$, i.e., $\delta(q_0,a)=\delta(q_0,b)=q_1$.
Considering the possible input words $01$ and $11$, thus we reach the same configuration $q_1 1 \$$, however, the former word should not be accepted, while the latter one is in $L^c$. By this contradiction, the proof has been finished.
\end{itemize}
\end{proof}

\begin{proposition}
  The language class accepted by SFAwtl is not closed under concatenation.
\end{proposition}
\begin{proof}
 Let us consider the languages $a+aaa$ and $b$ which both are accepted by state-deterministic FAwtl. Let us consider now their concatenation $L_c = \{ab,aaab\}$.
 Let us assume that there is an SFAwtl $A$ that accepts $L_c$.
 In its initial state $q_0$, there are the following options:
 \begin{itemize}
   \item There is no translucent letters for $q_0$, then there must be a transition with $a$ to a state $q_1 (\ne q_0)$ and no other transition from $q_0$. Neither $q_0$, nor $q_1$ is a final state. Now, at $q_1$ $A$ should be able to read a $b$ and it must reach an accepting state $q_2$. Now there are two subcases:
       \begin{itemize}
         \item If there is no translucency used at $q_1$, it must also have a transition with $a$ (to $q_2$) allowing to process the word $aaab$, however, in this case there would be a false acceptance of $aa$ by $A$ arriving to a contradiction.
         \item In the second subcase, $a$ must be translucent in $q_1$, and thus from the original input $aaab$, the remaining was $aab$ and in this way the last letter $b$ could be read, and then the remaining $aa$ should be accepted. However, in this case $A$ would also accept the word $abaa$ with a similar computation as $aaab$ contradicting to the fact that it accepts $L_c$.
       \end{itemize}
   \item If $a$ is translucent at $q_0$, then we must have a transition with $b$, then from $q_1$ $A$ should able to accept the remaining word $a$. However, in this case a similar computation would accept also $ba$ as the computation for $ab$. Contradiction.
   \item If $b$ is translucent at $q_0$, then we must have a transition with $a$, leading to the state $q_1$, similarly as in the first case. For $q_1$ the proof works in exactly in the same way as in that case.
   \item If both $a$ and $b$ are translucent in $q_0$, no transition can be defined, consequently, the nonempty language $L_c$ cannot be accepted in this way.
 \end{itemize}
\end{proof}

\section{The new models with nondeterministic translucency}\label{sec:def}

In this section, first we provide the formal definition of the new automata models and their work.

A \emph{(deterministic) finite state automaton with nondeterministically translucent letters}, abbreviated as (DFAwntl) {NFAwntl},
is defined as a septuple $A = (Q,\Sigma,\$,\tau,I,F,\delta)$,
 similarly to NFAwtl (DFAwtl), respectively, but without the condition
that
$\forall q\in Q\; \forall a\in \tau(q): \delta(q,a)=\emptyset.$ That is, for an NFAwntl
(DFAwntl) it is allowed that for a state $q\in Q$ and for a letter $a\in \Sigma$ both
$a\in \tau(q)$ and $\delta(q,a) \ne \emptyset$ hold.
Notice that at a DFAwntl, there is only one initial state, and there is at most one transition defined for any input letter, as at DFAwtl.

A computation step of $A$ is defined as follows. Assume that $w= a_1a_2\cdots a_n$ for some $n\ge 1$ and $a_1,\ldots,a_n\in\Sigma$ and $A$ is in state $q$.
Then $A$ looks for an occurrence
 of a letter (say $a_i = b$) for which a transition is defined, i.e., $\delta(q,b) \ne \emptyset$ such that
 each letter $a_j \in \tau(q)$ with $j<i$.
In this way, the actual configuration can be written as $q\cdot u b v \$ $ with letter $b$ in the position of $a_i$, $u\in \tau(q)^*$, $v\in \Sigma^*$, and the next configuration could be $p\cdot uv \$ $ for a state $p \in \delta(q,b)$. \\
 On the other hand, it may happen that such letter $a_i$ does not exist,
  i.e., there is a letter $a_i = c \not\in \tau(q)$ such that for each $j<i$ $a_j\in\tau(q)$ and $\delta(q,a_k)=\emptyset$ (for all $k\le j$)
  including $\delta(q,c)=\emptyset$. 
In this case $A$ halts without accepting; this computation fails. \\
Further, if $w\in(\tau(q))^*$ for a configuration $q w\$ $ with $q \in F$,
then $A$ reaches the \$-symbol and the computation halts by accepting. \\
Finally, $w\in(\tau(q))^*$, $q\not\in F$ and there is no letter in $w$ for which a transition has been defined ($\delta(q,a_i) = \emptyset$ for each $i$, or $w=\lambda$),
then the computation fails: $A$ does not accept.

A word $w\in\Sigma^*$ is \emph{accepted by} $A$ if there exists an initial state $q_0\in I$
and an accepting computation from $q_0w\cdot \$ $.
Now $L(A) = \{\,w\in\Sigma^*\mid w\mbox{ is accepted by }A\,\}$
is the \emph{language accepted by} $A$.

Based on these definitions, we can define four categories of NFAwtl:

\begin{center}
\begin{tabular}{l|cc|}
translucency $\backslash$ transition mapping & deterministic & nondeterministic \\
  \hline
disjoint & DFAwtl & NFAwtl \\
nondeterministic & DFAwntl & NFAwntl \\
  \hline
\end{tabular}
\end{center}

As we will show although the model DFAwntl seems deterministic by its transition function, we may easily cheat by the nondeterminism allowed by translucency.

\begin{example}\label{ex-DFAwntl}
  Let $A = (Q,\Sigma,\$,\tau,I,F,\delta)$, where $Q = \{q_0,q_a,q_b,q_c\}$,
$I=\{q_0\}$, $F=\{q_c\}$, $\Sigma = \{a,b,c\}$,
and the functions $\tau$ and $\delta$ are defined as follows:
$$\begin{array}{rclrcl}
\tau(q_0) & = & \emptyset, \hspace{2.5cm}      & \delta(q_0,a) & = & \{q_a\}, \delta(q_0,b)  =  \{q_b\}, \delta(q_0,c)  =  \{q_c\}, \\
\tau(q_a) & = & \{a,b,c\},        & \delta(q_a,b) & = & \{q_0\},\\
\tau(q_b) & = & \{a,b,c\},        & \delta(q_b,a) & = & \{q_0\},\\
\tau(q_c) & = & \emptyset.        & 
\end{array}$$
Further, $\delta(q,x)=\emptyset$ for all other pairs $(q,x)\in Q\times\Sigma$.
It is easy to check that $A$ is a DFAwntl. \\
Let us see how $A$ works. Since there are no translucent letters in $q_0$, $A$ consumes the first letter of the remaining input always in this state. If it was an $a$, then it erases a $b$ from anywhere in the tape; if $A$ consumes a $b$ at state $q_0$, then in the next computation step $A$ erases an $a$ from anywhere in the tape. Finally, the input is accepted if only a $c$ remains on the tape  and $A$ is in state $q_0$, then it reaches the accepting state $q_c$.

Thus, for every accepted word the number of its $a$-s and $b$-s are the same and it contains a $c$. Let us write such a word in the form $v c u$ with $v,u \in \{a,b\}^*$.
It is also easy to see that $A$ may accept various words where $|v| \ge |u|$, but no words with $|v| < |u|$.
On the other hand, let us see which words are accepted with the property $|v| = |u|$.
By the work of $A$, the conditions $|v|_a = |u|_b$ and $|v|_b = |u|_a$ must hold, i.e., in $v$ the number of $a$-s is the same as the number of $b$-s in $u$ and vice versa.

Now we are arguing that the same language cannot be accepted by any NFAwtl without using nondeterministic translucency (due to lack of space we skip some parts of the formal proof).
Let us assume that there is an NFAwtl $B$ that accepts the same language as $A$. Let the number of states of $B$ is $n$.
Let us consider a word $w = a^kb^\ell c a^\ell b^k \in L(A)$ with $k,\ell >2n$. By our assumption $B$ accepts $w$, thus consider an accepting computation on $w$ by $B$. Clearly, there are two cases based on the first $n+1$ steps of the computation.
\begin{itemize}
\item  If letter $c$ is erased during these computation step, then it can be shown that some words $v c u$ with $|v| < |u|$ would also be accepted by $B$ having all the letters processed after the step in which $c$ is read after the $c$ in the original input word.
\item If $c$ is not read during the first $n+1$ steps, only $a$-s and $b$-s before the $c$ (in part $v$) are accessed and processed in the first $n+1$ steps. However, by the pigeon-hole principle, a state is repeated during these steps, meaning that
there are also values $i$ and $j$ ($0\leq i,j\leq n+1$, $i+j>0$) such that
 from input $w' = a^{k-i}b^{\ell-j} c a^\ell b^k$ in $n+1-i-j$ steps exactly the same configuration is reached as from $w$ in $n+1$ steps.
 In this case, by continuing the computation on $w'$ in the same way as the accepting computation on $w$, the word $w'$ will also be accepted.
\end{itemize}
Now, in both cases, we have reached contradiction by accepting words of the form $v c u$ for which $|v| < |u|$.
\end{example}

\subsection{Hierarchy of the accepted languages}

In this subsection our aim is  to
 give some
hierarchy like results by establishing where the new families of languages are comparing them to various other classes.

First, we note that, in fact, the following inclusions hold by definition.
\begin{proposition}
Every NFAwtl is an NFAwntl and every DFAwtl is a DFAwntl. \\ Moreover, every DFAwntl is an NFAwntl.
\end{proposition}

Based on Example~\ref{ex-DFAwntl}, %
we can also state some hierarchy results on the accepted language classes.
\begin{proposition}
The language class accepted by  NFAwtl is a proper subclass of the language class accepted by NFAwntl. \\
The language class accepted by  DFAwtl is a proper subclass of the language class accepted by DFAwntl. \\
\end{proposition}

Here, we leave open the question if NFAwntl is more efficient and expressive than DFAwntl.

 On the one hand, we have seen that we can construct DFAwntl that accept some non context-free languages. Now, on the other hand, let us show some of their limitations.

To compare the new language classes with some classical classes of formal languages we establish the following result.

\begin{lemma}
  For every language accepted by an NFAwntl (DFAwntl), there is a letter equivalent sublanguage that is accepted by an NFAwtl (DFAwtl, resp.).
\end{lemma}
\begin{proof}
In case there is no such letter in any state which is both in the set of translucent letters and there is also a transition on it, the automaton is in fact, an NFAwtl (also a DFAwtl in deterministic case) and its language, as its own sublanguage, fulfils the statement of the lemma.

Now, let us assume that automaton $A$ is an NFAwntl, but it is not an NFAwtl.
 On the one hand,  since there is no ``forced'' way not to see a letter for which a transition is defined, $A$ may always consumes the first occurrence of such a letter and in fact accepts words of a language that is also accepted by an NFAwtl.
More precisely, by removing those letters from the translucency set of a state for which transitions are defined, an NFAwtl (DFAwtl) $A'$ can be obtained.
Clearly all words that $A'$ may accept are also accepted by $A$ with a similar computation.

 To see that this language is letter equivalent to the originally accepted language, consider a computation on any of the accepted word by $A$.
Since the automaton never knows if the consumed letter in a computation step is the first or only reached through some translucent letters, we may reorder the input according to an accepting computation, and in this way for each accepted word there will be a letter equivalent word that has also been accepted by $A'$.
\end{proof}

Moreover, if, by chance, the letters of the input are ordered in exactly the way as they are consumed during an accepting computation, then, in fact, an NFAwntl is working in the same way as an NFA, thus we may also establish the following fact:
\begin{lemma}
  For every language accepted by an NFAwntl, there is a letter equivalent regular sublanguage.
\end{lemma}

From the previous lemma we can conclude:
\begin{corollary}
  All languages accepted by NFAwntl are semi-linear.
\end{corollary}

This could be interesting in the mirror of the fact, that by changing the window size of an R automata from 1 to 2 (like allowing to have the translucency and transitions not  letter by letter, but somehow by pairs of consecutive letters), the corresponding model, the CD system of stateless deterministic R(2) automata is able to accept some non semi-linear languages \cite{jalcR(2)}.

On the other hand, as the linear context-free language $\{a^nb^n\}$ does not have any letter equivalent regular sublanguage, there is no NFAwntl that could accept it.
 Since $\{a^nb^n\}$ is deterministic linear and also in 2detLIN (accepted by deterministic 2-head finite automata consuming the input letters from the two extremes until they are meeting \cite{ActInf-Sh}), we have the following incomparability results.

\begin{theorem}
  The language classes accepted by NFAwntl and DFAwntl properly include the class of regular languages. Further,
  the language classes accepted by NFAwntl and DFAwntl are incomparable to each of the following classes of languages: deterministic linear, 2detLIN, linear context-free, deterministic context-free, context-free.
\end{theorem}

Finally, we analyse the computations of DFAwntl. 

\subsection{Simulating nondeterministic computations with DFAwntl}

In this section our aim is to %
show that although seemingly by the transition function, DFAwntl seem to be deterministic
automata, they have some real nondeterministic features.

\begin{example}\label{exa6}
  As we mentioned in Example~\ref{exa:nmkj}, there are DFAwtl that accept the languages
  $$L^{ab} = \{w\in\{a,b,c,d\}^*~|~ |w|_a = |w|_b >0 \text{ and } |w|_c = |w|_d  \}, $$
  $$L^{ac} = \{w\in\{a,b,c,d\}^*~|~ |w|_a = |w|_c > 0 \text{ and } |w|_b = |w|_d \}$$ and
  $$L^{ad} = \{w\in\{a,b,c,d\}^*~|~ |w|_a = |w|_d >0 \text{ and } |w|_b = |w|_c \}.$$
  Let these automata be having the set of states $Q^{ab} = \{q_0^{ab},q_a^{ab},q_b^{ab},q_c^{ab},q_d^{ab}, q^{ab}\}$;   $Q^{ac} =\{q_0^{ac},q_a^{ac},q_b^{ac},q_c^{ac}, \\ q_d^{ac}, q^{ac}\}$ and
    $Q^{ad} =\{q_0^{ad},q_a^{ad},q_b^{ad},q_c^{ad},q_d^{ad}, q^{ad}\}$, respectively.

  The union of these languages cannot be accepted by any DFAwtl.
  Let us define now a DFAwntl as follows.

 $A = (Q,\Sigma,\$,\tau,I,F,\delta)$, where $Q = Q^{ab} \cup Q^{ac} \cup Q^{ad} \cup \{q_0\}$,
$I=\{q_0\}$, $F=\{q^{ab},q^{ac},q^{ad}\}$, $\Sigma = \{a, b, c,\\ d\}$,
and the functions $\tau$ and $\delta$ are defined as follows:

$$\tau(q_0)  =  \{a,b,c,d \}, \hspace{1.5cm}     \delta(q_0,a) = \emptyset,  \delta(q_0,b)  =  \{q_b^{ab}\},
\delta(q_0,c) = \{q_c^{ac}\}, \delta(q_0,d) = \{q_d^{ad}\}
$$ for the newly added initial state and they are inherited from the respective automata for
 each other state.

Now, for any input, $A$ nondeterministically guesses in which of the three languages the input is: as $\tau(q_0) = \Sigma$ it has access to any occurrence of a $b$, a $c$ or a $d$ as transitions are defined for these letters in $q_0$.
 By guessing the input belonging to $L^{ab}$, it should have at least one $b$, thus by reading a $b$ in the first step, the subautomaton accepting $L^{ab}$ is chosen such that a $b$ has already been processed. Now, it is easy to see that after this nondeterministic choice in the first step, the computation continues in a deterministic manner.
The other nondeterministic choices in the first step of the computation are: by consuming a $c$ anywhere from the input $A$ chooses to check whether the input belongs to $L^{ac}$, and by consuming a $d$ anywhere from the input in the first step of the computation, $A$ chooses to check whether the input belongs to $L^{ad}$.
 If the guess was correct, the input will be accepted. Otherwise, $A$ must use another computation to accept the given input, if any.  Consequently the DFAwntl $A$ accepts
 $L^{ab} \cup L^{ac} \cup L^{ad}$.
\end{example}

We left open if NFAwntl can accept more languages than DFAwntl (the properness of the inclusion relation of these two language classes is open).
It is known (see \cite{IJCM}, for the proof) that DFAwtl cannot accept all rational trace languages, and this fact was used to prove the properness of the hierarchy between NFAwtl and DFAwtl. Actually, it was shown that the language $\{ w\in\{a,b\}^*~|~ |w|_a = |w|_b \text{ or } 2|w|_a = |w|_b \}$ cannot be accepted by any DFAwtl.
Here we show that with a similar method as in Example \ref{exa6}, DFAwntl is able to accept this language. The automaton is shown in Figure~\ref{figure:Hier}. In each state the indicated set, if any, shows the translucent letters of the given state. The other notation is standard, e.g., double circles for final states, etc.

\begin{figure}[!ht] 
  \centering
  \includegraphics[scale=0.85]{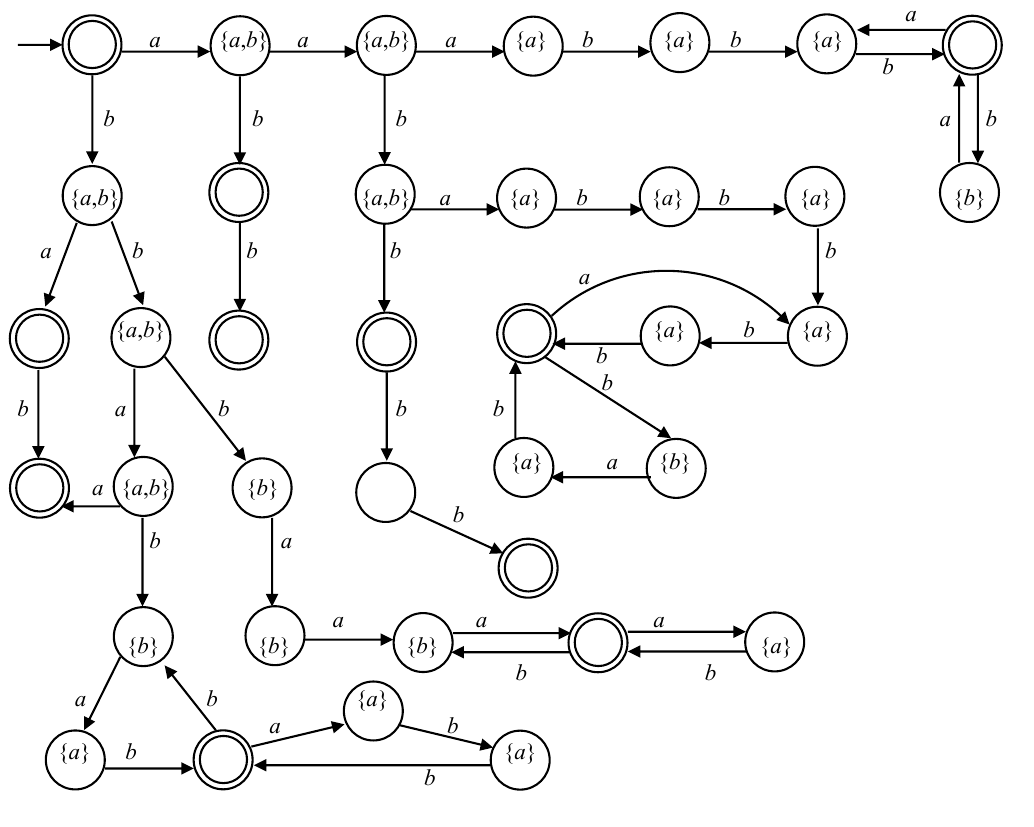}
  \caption{A DFAwntl accepting $\{ w\in\{a,b\}^*~|~ |w|_a = |w|_b \text{ or } 2|w|_a = |w|_b \}$.}\label{figure:Hier}
\end{figure}

Last, but not least, we present some closure properties.

\begin{proposition} The class of
languages accepted by NFAwntl is closed under union. \end{proposition}
\begin{proof} Wlog. we may assume that the two languages are over the same alphabet $\Sigma$. Now,
having two NFAwntl, say $A_1 = (Q_1,\Sigma,\$,\tau_1,I_1,F_1,\delta_1)$ and $A_2 = (Q_2,\Sigma,\$,\tau_2,I_2,F_2,\delta_2)$ with $Q_1 \cap Q_2 = \emptyset $, we construct $A = (Q_1\cup Q_2,\Sigma,\$,\tau,I_1\cup I_2,F_1\cup F_2,\delta)$, where $\tau(q) =
 \left\{
         \begin{array}{ll}
           \tau_1(q), & \hbox{if $q\in Q_1$;} \\
           \tau_2(q), & \hbox{if $q\in Q_2$.}
         \end{array}
       \right.
$ and for each $a\in\Sigma$, $\delta(q,a) =
 \left\{
         \begin{array}{ll}
           \delta_1(q,a), & \hbox{if $q\in Q_1$;} \\
           \delta_2(q,a), & \hbox{if $q\in Q_2$.}
         \end{array}
       \right.
$.
 Then $A$ may choose nondeterministically among the possible initial states,
depending on if the chosen state is in $I_1$ or $I_2$, $A$ will do a computation that is similar to a computation of $A_1$ or $A_2$, respectively. \end{proof}

It is known that the language class accepted by DFAwtl is not closed under union \cite{JCSS}. On the other hand,
we have also seen, that DFAwntl may also be able to compute the union of some languages accepted by DFAwntl, however, in general we leave open the problem if this class is closed under union.

On the other hand, the language class of NFAwtl is closed under concatenation, and the proof was based on guessing when the last occurrence of the letters are consumed to give a construction when the last occurrence of any letter was consumed without using any translucency \cite{JCSS}. As in the new model, generally, we may not be sure when the last occurrence is consumed (maybe even in the first step), the original construction definitely does not work. Thus, the closure of the new classes under concatenation is also left as an open problem.

The fact that each accepted language must have a letter equivalent regular sublanguage and the examples shown lead also to the following non-closure property:

\begin{proposition}
Language classes accepted by NFAwntl and DFAwtl are not closed under intersection with regular sets, and thus they are not closed under intersection.
\end{proposition}

\section{Conclusions} %

Recently, another extension of the finite automata with translucent letters was investigated in which in the computation the head is not restarting after erasing a symbol, but continues from the position where this letter has been erased (or by reaching the endmarker, it starts from the beginning again) \cite{Fran-Fred}. This model is defining some new interesting classes of languages that are superclasses of the classes of languages of the original model,  %
as the new model is able to simulate the original nondeterministic finite automata with translucent letters. %
Our extensions, the FAwntl, are such extensions that the original deterministic and nondeterministic finite automata with translucent letters are special cases of our new models (we may not need to simulate them as they are included in our new classes of automata), thus
the computational power of the original models has been increased by relaxing the condition
of disjointness of the sets of letters for a state which is containing the translucent letters of the given state and which is containing the letters that can be read in the given state. However, we left open if nondeterministic transitions are more powerful in case we allow nondeterministic translucency (the author guesses/conjectures that the model DFAwntl is weaker than NFAwntl in terms of accepted languages). It is also left open if DFAwntl are able to accept all rational trace languages.

Although the expressive power has been increased, the new model still has various limitations, as the accepted languages must always have a letter equivalent regular sublanguage. As the class is not closed under intersection with regular languages,
 transduced-input variants can be investigated and studied in the future
similarly to \cite{CRBulg}. Various closure properties of the new classes of languages are left open, they are also subjects of future studies.

We believe that the combination of the new directions by continuing the computation from the position of the erased letter and by using nondeterministic translucency, can fruitfully be considered also in the future.

In the other newly investigated model we have applied the state-deterministic restriction
for FAwtl showing that this model is accepting an interesting family of languages. Combining state-determinism and nondeterministic translucency and/or the non-returning feature could also be a nice topic for future research.

\section*{Acknowledgements} The comments of the anonymous reviewers are gratefully acknowledged.

\nocite{*}
\bibliographystyle{eptcs}
\bibliography{translucent_bibliography}
\end{document}